\documentclass[a4paper]{elsarticle}
\usepackage{amsmath}
\usepackage{amssymb}
\usepackage{amsthm}
\usepackage{stmaryrd}
\usepackage{tikz-cd}
\usetikzlibrary{positioning}
\usetikzlibrary{shapes.geometric}
\usepackage{hyperref}
\hypersetup{
    colorlinks,
    linkcolor={blue!70!black},
    citecolor={blue!70!black},
    urlcolor={blue!70!black}
}
\usepackage[textsize=tiny]%[disable]
{todonotes}
\usepackage{thmtools}
\usepackage{thm-restate}
\usepackage[mathscr]{eucal}
\usepackage{calc}
\usepackage{adjustbox}
\usepackage{wrapfig}
\usepackage{subfig}
\usepackage{cutwin}
\usepackage{microtype}

\let\counterwithin\relax
\usepackage{chngcntr}
\usepackage{apptools}
\AtAppendix{\counterwithin{proposition}{section}}

\newcommand{\Cat}[1]{\ensuremath{{\mathbf{#1}}}} % Categories go in bold font
\newcommand{\CatC}{\Cat{C}}
\newcommand{\Set}{{\Cat{Set}}} % Category of sets
\newcommand{\TAlt}{T_\mathtt{Alt}}

\theoremstyle{plain}
\newtheorem{theorem}{Theorem}[section]
\newtheorem{lemma}[theorem]{Lemma}
\newtheorem{proposition}[theorem]{Proposition}

\theoremstyle{definition}
\newtheorem{construction}[theorem]{Construction}
\newtheorem{definition}[theorem]{Definition}
\newtheorem{example}[theorem]{Example}

\theoremstyle{remark}

\opencutright

\newlength{\strutheight}
\newcommand{\id}{\mathsf{id}}
\newcommand{\nto}{\Rightarrow}

\newcommand{\Ps}{\mathcal{P}}

\newcommand{\F}{\mathbb{F}}

\newcommand{\eword}{\varepsilon}

\newcommand*{\smallcircled}[1]{{\tikz[baseline=(X.base)]\node(X)[draw,shape=circle,inner sep=0]{\text{\scriptsize\strut$#1$}};}}

\title{A (co)algebraic theory of succinct automata\tnoteref{grants}}
\tnotetext[grants]{This work was partially supported by ERC starting grant ProFoundNet (679127) and a Leverhulme Prize (PLP-2016-129).}

\author[1]{Gerco van Heerdt}
\author[1,2]{Joshua Moerman}
\author[1]{Matteo Sammartino}
\author[1]{Alexandra Silva}
\address[1]{University College London}
\address[2]{Radboud University}
\usepackage{tikz}
\usetikzlibrary{cd,arrows,automata}
\tikzset{every state/.style={minimum size=0pt}}

\begin{document}

\begin{abstract}
	The classical subset construction for non-deterministic automata can be generalized to other side-effects captured by a monad.
	The key insight is that both the state space of the determinized automaton and its semantics---languages over an alphabet---have a common algebraic structure: they are Eilenberg-Moore algebras for the powerset monad.
	In this paper we study the reverse question to determinization.
	We will present a construction to associate succinct automata to languages based on different algebraic structures.
	For instance, for classical regular languages the construction will transform a deterministic automaton into a non-deterministic one, where the states represent the join-irreducibles of the language accepted by a (potentially) larger deterministic automaton.
	Other examples will yield alternating automata, automata with symmetries, CABA-structured automata, and weighted automata.
\end{abstract}

\maketitle

\section{Introduction}
Non-deterministic automata are often used to provide compact representations of regular languages. Take, for instance, the language 
\[
	\mathcal L = \{ w\in\{a,b\}^* \mid |w|>2 \text{ and the $3^{\mathrm{rd}}$ symbol from the right is an $a$}\} .
\]
There is a simple non-deterministic automaton accepting it (below, top automaton) and it is not very difficult to see that the smallest deterministic automaton (below, bottom automaton) will have $8$ states.
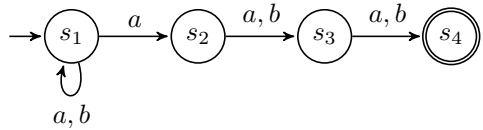
\begin{center}\vspace{-.3cm}
\begin{tikzpicture}[initial text={},->,>=stealth',shorten >=1pt,auto,node distance=11ex,semithick]
\node[initial,state] (0) {$s_1$};	
\node[state] (1) [right of = 0] {$s_2$};
\node[state] (2) [right of = 1] {$s_3$};
\node[state,accepting] (3) [right of = 2] {$s_4$};

\path 
(0) edge node {$a$} (1)
(0) edge[loop below] node {$a,b$} (0)
(1) edge node {$a,b$} (2)
(2) edge node {$a,b$} (3);
\end{tikzpicture}
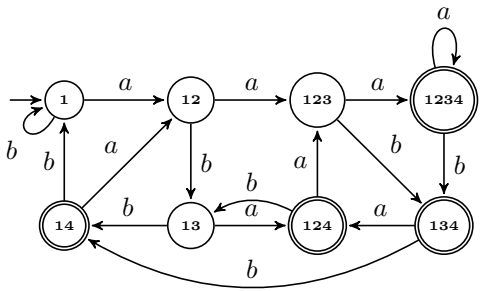
\vspace{-.8cm}
\begin{tikzpicture}[initial text={},->,>=stealth',shorten >=1pt,auto,node distance=11ex,semithick]
\node[initial,state] (0a) [right of = 3, node distance=16ex] {\tiny$\bf 1$};	
\node[state] (1a) [right of = 0a] {\tiny$\bf 12$};
\node[state] (2a) [right of = 1a] {\tiny$\bf 123$};
\node[state,accepting] (3a) [right of = 2a] {\tiny$\bf 1234$};
\node[state,accepting] (4) [below of = 0a] {\tiny$\bf 14$};
\node[state] (5) [below of = 1a] {\tiny$\bf 13$};
\node[state,accepting] (6) [below of = 2a] {\tiny$\bf 124$};
\node[state,accepting] (7) [below of = 3a] {\tiny$\bf 134$};

\path 
(0a) edge node {$a$} (1a)
(0a) edge[out=240,in=200,looseness=8] node {$b$} (0a)
(1a) edge node {$a$} (2a)
(2a) edge node {$a$} (3a)
(3a) edge[loop above] node {$a$} (0a)
(6) edge node {$a$} (2a)
(7) edge [swap] node {$a$} (6)
(7) edge [bend left, swap] node {$b$} (4)
(6) edge [bend right] node [swap] {$b$} (5)
(3a) edge node {$b$} (7)
(2a) edge node {$b$} (7)
(4) edge node {$a$} (1a)
(4) edge node {$b$} (0a)
(5) edge node {$a$} (6)
(5) edge [swap] node {$b$} (4)
(1a) edge node {$b$} (5)
;
\end{tikzpicture}	
\end{center}\vspace{-.2cm}
The labels we chose for the states of the deterministic automaton are not coincidental---they represent the subsets of states of the non-deterministic automaton that would be obtained when constructing a deterministic one using the classical subset construction.

The question we want to study in this paper has as starting point precisely the observation that non-deterministic automata provide compact representations of languages and hence are more amenable to be used in algorithms and promote scalability. In fact, the origin of our study goes back to our own work on automata learning~\cite{moerman2017}, where we encountered large nominal automata that, in order for the algorithm to work for more realistic examples, had to be represented non-deterministically. In other recent work~\cite{bollig2009,angluin2015}, different forms of non-determinism are used to learn compact representations of regular languages. This left us wondering whether other {\em side-effects} could be used to overcome scalability issues.

Moggi~\cite{DBLP:journals/iandc/Moggi91} introduced the idea that {\em monads} could be used a general abstraction for side-effects. A monad is a triple $(T,\eta,\mu)$ in which $T$ is an endofunctor over a category whose objects can be thought of as capturing pure computations. The monad is equipped with a unit $\eta \colon X\to TX$, a natural transformation that enables embedding any pure computation into an effectful one, and a multiplication $\mu \colon TTX \to TX$ that allows flattening nested effectful computations. Examples of monads capturing side-effects include powerset (non-determinism) and distributions (randomness).

Monads have been used extensively in programming language semantics (see e.g.~\cite{Swamy:2011:LMP:2034773.2034778} and references therein). More recently, they were used in categorical studies of automata theory~\cite{DBLP:conf/dlt/Bojanczyk15}. One example of a construction in which they play a key role is a generalization of the classical subset construction to a class of automata~\cite{SilvaBBR13,DBLP:conf/fsttcs/SilvaBBR10}, which we will describe next. 

The classical subset construction, connecting non-deterministic and deterministic automata, can be described concisely by the following diagram.
\begin{align*}
	\begin{tikzcd}[column sep=.7cm,row sep=.7cm,ampersand replacement=\&]
		X \ar{d}[swap]{\delta} \ar{r}{\{-\}} \&
			\mathcal{P}(X) \ar{dl}{\delta^\sharp} \ar[dashed]{r}{l} \&
			2^{A^*} \ar{d}{<\epsilon?,\partial>} \\
		2 \times \mathcal{P}(X)^A \ar[dashed]{rr}{\id \times l^A} \&
			\&
			2 \times (2^{A^*})^A
	\end{tikzcd}
\end{align*}	
	 We omit initial states and represent a non-deterministic automaton as a pair $(X,\delta)$ where $X$ is the state space and $\delta \colon X \to 2\times \mathcal{P}(X)$  is the transition function which has in the first component the (non-)final state classifier.  The language semantics of a non-deterministic automaton $(X, \delta)$ is obtained by first constructing a deterministic automaton $( \mathcal{P}(X), \delta^\sharp)$ which has a larger state space consisting of subsets of the original state space and then computing the accepted language of the determinized automaton. The language map $l$ associating the accepted language to a state is a universal map: for every deterministic automaton $(Q, Q \to 2\times Q^A)$ the map $l$ is the unique map into the automaton of languages $(2^{A^*}, 2^{A^*} \xrightarrow{<\epsilon?,\partial>} 2 \times (2^{A^*})^A)$.
	
The universal property of the automaton of languages inspired the development of a categorical generalization of automata theory, including of the subset construction which we detail below. In particular, we can consider {\em general} automata as pairs $(X, X \xrightarrow{t} FX)$ where the transition dynamics $t$ is parametric on a functor $F$. Such pairs are usually called coalgebras for the functor $F$~\cite{jan-universal}.  For a wide class of functors $F$, the category of coalgebras has a final object $(\Omega, \omega)$, the so-called \emph{final} coalgebra, which plays the analogue role to languages.

The classical subset construction was generalized in previous work~\cite{SilvaBBR13} by replacing deterministic automata with coalgebras for a functor $F$ and the powerset monad with a suitable monad $T$. As above, it can be summarized in a diagram:
\begin{align*}
		\begin{tikzcd}[column sep=.7cm,row sep=.7cm,ampersand replacement=\&]
			X \ar{d}[swap]{\delta} \ar{r}{\eta} \&
			TX \ar{dl}{\delta^\sharp} \ar[dashed]{r}{l} \&
				\Omega \ar{d}{\omega} \\
			FTX \ar[dashed]{rr}{Fl} \&
				\&
				F\Omega
		\end{tikzcd}
\end{align*}
The monad $T$ will be the structure we will explore to enable succinct representations. 
The crucial ingredient in generalizing the subset construction was the observation that the target of the transition dynamics---$2\times \Ps(-)^A$---and the set of languages---$2^{A^*}$---both have a complete join-semilattice structure. This enables one to define the determinized automaton as a unique lattice extension of the non-deterministic one, and, moreover, the language map $l$ preserves the semantics: $l(\{s_1,s_2\}) = l(\{s_1\})\cup l(\{s_2\})$.

This latter somewhat trivial observation was also exploited in the work of Bonchi and Pous~\cite{BonchiP15} in defining an efficient algorithm for language equivalence of NFAs by using coinduction-up-to.
Join-semilattices are precisely the Eilenberg-Moore algebras of the powerset monad, and one can show that if a functor has a final coalgebra in $\Set$, this can be lifted to the category of Eilenberg-Moore algebras of a monad $T$ ($T$-algebras).
This makes it possible to construct the more general diagram above, where the coalgebra structure is generalized using a functor $F$ and a monad $T$. The only assumptions for the existence of $T$-algebra maps $\delta^\sharp$ and $l$ are the existence of a final coalgebra for $F$ in $\Set$ and that $FTX$ can be given a $T$-algebra structure.

In this paper we ask the reverse question---given a deterministic automaton, if we assume the state space has a join-semilattice structure, can we build a corresponding succinct non-deterministic one?
More generally, given an $F$-coalgebra in the category of $T$-algebras, can we build a succinct $FT$-coalgebra in the base category that represents the same behavior?

We will provide an abstract framework to understand this construction, based on previous work by Arbib and Manes~\cite{arbib1975_}.
Our abstract framework relies on alternative, more modern, presentation of some of their results.
Due to our focus on set-based structures, we will conduct our investigation within the category $\Set$, which enables us to provide effective procedures.
This does mean that not all of the results due to Arbib and Manes will be given in their original generality.
We present a comprehensive set of examples that will illustrate the versatility of the framework.
We also discuss more algorithmic aspects that are essential if the present framework is to be used as an optimization, for instance as part of a learning algorithm.

After recalling basic facts about monads and structured automata in Section~\ref{sec:preliminaries}, the rest of this paper is organized as follows:
\begin{itemize}
	\item In Section~\ref{sec:succinct} we introduce a general notion of \emph{generators} for a $T$-algebra, and we show that automata whose state space form a $T$-algebra---which we call $T$-automata---admit an equivalent \emph{$T$-succinct automaton}, defined over generators. We also characterize \emph{minimal} generators and give a condition under which they are globally minimal in size.
	\item In Section~\ref{sec:minimization} we give an effective procedure to find a minimal set of generators for a $T$-algebra, and we present an algorithm that uses that procedure to compute the $T$-succinct version of a given $T$-automaton. The algorithm works by first minimising the $T$-automaton: the explicit algebraic structure allows states that correspond to algebraic combinations of other states to be detected, and then discarded when generators are computed.
 	\item In Section~\ref{sec:main} we show how the algorithm of Section~\ref{sec:minimization} can be applied to ``plain'' finite automata---without any algebraic structure---in order to derive an equivalent $T$-succinct automaton. We conclude with a result about the compression power of our construction: it produces an automaton that is \emph{at least as small} as the minimal version of the original automaton.
 
	\item Finally, in Section~\ref{sec:examples} we give several examples, and in Section~\ref{sec:conclusions} we discuss related and future work.
\end{itemize}

\section{Preliminaries}\label{sec:preliminaries}
Side-effects and different notions of non-determinism can be conveniently captured as a \emph{monad} $T$ on a category \CatC.   A \emph{monad} $T = (T,\mu,\eta)$ is a triple consisting of an
endofunctor $T$ on $\CatC$ and two natural transformations: a
\emph{unit} $\eta\colon \mathtt{Id} \Rightarrow T$
and a \emph{multiplication} $\mu\colon
T^2 \Rightarrow T$.
They satisfy the following laws:
%\vspace*{-.3cm}
\begin{align*}
	\mu \circ \eta{T} = \id = \mu \circ T\eta &
		&
		\mu \circ \mu{T} = \mu \circ T\mu.
\end{align*}
An example is the triple $(\Ps, \{-\}, \bigcup)$ where $\Ps$ denotes the powerset functor in $\Set$ that assigns to each set the set of all its subsets, $\{-\}$ is the function that returns a singleton set, and $\bigcup$ is just union of sets.

Given a monad $T$, the category of  $\CatC^T$ of Eilenberg-Moore algebras over $T$, or simply $T$-algebras, has as objects pairs $(X,h)$ consisting of an object $X$, called
carrier, and a morphism $h\colon TX \rightarrow X$  such that $h \circ \mu_X = h \circ Th$
and $h \circ \eta_X = \id_X$. A $T$-homomorphism between two $T$-algebras $(X,h)$ and $(Y,k)$
is a morphism $f\colon X \to Y$ such that $f \circ h = k \circ Tf$. 

We will often refer to a $T$-algebra $(X, h)$ as $X$ if $h$ is understood or if its specific definition is irrelevant.
Given an object $X$, $(TX, \mu_X)$ is a $T$-algebra called the \emph{free $T$-algebra} on $X$. Given an object $U$ and a $T$-algebra $(V, v)$, there is a bijective correspondence between $T$-algebra homomorphisms $TU \to V$ and morphisms $U \to V$: for a $T$-algebra homomorphism $f \colon TU \to V$, define $f^\dagger = f \circ \eta \colon U \to V$; for a morphism $g \colon U \to V$, define $g^\sharp = v \circ Tg \colon TU \to V$.
Then $g^\sharp$ is a $T$-algebra homomorphism called the \emph{free $T$-extension of $g$}, and we have
\begin{align}\label{eq:sharpinv}
	f^{\dagger\sharp} = f &
		&
		g^{\sharp\dagger} = g.
\end{align}
Furthermore, for all objects $S$ and morphisms $h \colon S \to U$,
\begin{equation}\label{eq:sharpcomp}
	g^\sharp \circ Th = (g \circ h)^\sharp.
\end{equation}

\begin{example}
	For the monad $\Ps$ the associated Eilenberg-Moore category is the category of (complete) join-semilattices. Given a set $X$, the free $\Ps$-algebra on $X$ is the join-semilattice $(\Ps X,\bigcup)$ of subsets of $X$ with the union operation as join.
\end{example}

Although some results are completely abstract, the central definition of minimal generators in Section~\ref{sec:succinct} is specific to monads $T$ on the category $\Set$.
Therefore we restrict ourselves to this setting.
More precisely, we consider automata over a finite alphabet $A$ with outputs in a set $O$.
In order to define automata in $\Set^T$ as (pointed) coalgebras for the functor $O \times (-)^A$, we need to lift this functor from $\Set$ to $\Set^T$.
Such a lifting corresponds to a \emph{distributive law} of $T$ over $O \times (-)^A$~\cite[see e.g.,][]{johnstone1975}.
A distributive law of the monad $T$ over a functor $F \colon \Set \to \Set$ is a natural transformation $\rho \colon TF \nto FT$ satisfying $\rho \circ \eta{F} = F\eta$ and $F\mu \circ \rho{T} \circ T\rho = \rho \circ \mu{F}$.
In most examples we will define a $T$-algebra structure $\beta \colon TO \to O$ on $O$, which is well known to induce a distributive law $\rho \colon T(O \times (-)^A) \nto O \times T(-)^A$ given by
\begin{equation}\label{eq:stddl}
	\rho_X = T(O \times X^A) \xrightarrow{\langle{}T\pi_1, T\pi_2\rangle} TO \times T(X^A) \xrightarrow{\beta \times \rho'_X} O \times T(X)^A
\end{equation}
for any set $X$, where $\rho'(U)(a) = T(\lambda f \colon A \to X. f(a))$.
In general, we assume an arbitrary distributive law $\rho \colon T(O \times (-)^A) \nto O \times T(-)^A$, which gives us the following notion of automaton.
\begin{definition}[$T$-automaton]
	A \emph{$T$-automaton} is a triple $(X, i \colon 1 \to X, \delta \colon X \to O \times X^A)$, where $X$ is an object of $\Set^T$ denoting the state space of the automaton, $i$ is a function designating the initial state, and $\delta$ is a $T$-algebra map assigning an output and transitions to each state.
\end{definition}
Notice that the initial state map $i \colon 1 \to X$ in the above definition is not required to be a $T$-algebra map.
However, it corresponds to the $T$-algebra map $i^\sharp \colon T1 \to X$.
Thus, a $T$-automaton is an automaton in $\Set^T$.

The functor $F(X) = O \times X^A$ has a final coalgebra in $\Set^T$~\cite{jacobs2012} that can be used to define the \emph{language} accepted by a $T$-automaton.

\begin{definition}[Language accepted]\label{def:language}
	Given a $T$-automaton $(X, i\colon 1\to X, \delta\colon X \to O \times X^A)$, the language accepted by $X$ is $l \circ i \colon 1 \to O^{A^*}$, where $l$ is the final coalgebra map.
	In the diagram below, $\omega$ is the final coalgebra.
	\[
		\begin{tikzcd}
			1 \ar{dr}{i} \\[-20pt]
			&
				X \ar{dd}{\delta} \ar[dashed]{r}{l} &
				O^{A^*} \ar{dd}{\omega} &
				\omega(\varphi) = (\varphi(\eword), \lambda a\;w. \varphi(aw)) \\[-20pt]
			&
				&
				&
				l(x)(\eword) = \pi_1(\delta(x)) \\[-20pt]
			&
				O \times X^A \ar[dashed]{r}{\id \times l^A} &
				O \times (O^{A^*})^A &
				l(x)(aw)= l(\pi_2(\delta(x))(a))(w)
		\end{tikzcd}
	\]
	We use $\eword$ to denote the empty word.
\end{definition}

If the monad $T$ is finitary, then the category $\Set^T$ is locally finitely presentable, and hence it admits (strong epi, mono)-factorizations~\cite{adamek1994}.
As in \cite{arbib1975_}, we use these factorizations to quotient the state-space of an automaton under language equivalence.
The transition structure, $\gamma$, is obtained by diagonalization via the factorization system.
Diagramatically:
\begin{equation}\label{eq:obs}
	\begin{tikzcd}[column sep=1.0cm,row sep=.7cm]
		1 \ar{d}[swap]{i} \ar{dr}{j} \\
		X \ar{d}[swap]{\delta} \ar[two heads]{r}{e} &
			M \ar[dashed]{d}{\gamma} \ar[tail]{r}{m} &
			O^{A^*} \ar{d}{\omega} \\
		O \times X^A \ar{r}{\id \times e^A} &
			O \times M^A \ar[tail]{r}[pos=.6]{\id \times m^A} &
			O \times (O^{A^*})^A
	\end{tikzcd}
\end{equation}
Here the epi $e$ and mono $m$ are obtained by factorizing the final coalgebra map $l \colon X \to O^{A^*}$.
We call the quotient automaton $(M,j,\gamma)$ the \emph{observable quotient} of $(X,i,\delta)$.

\section{$T$-succinct automata}\label{sec:succinct}

Given a $T$-automaton $\mathcal{X} = (X,i,\delta)$,
our aim is to obtain an equivalent automaton in $\Set$ with transition function $Y \to O \times T(Y)^{A}$, where $Y$ is smaller than $X$.\footnote{%
	Here, we are abusing notation and using $O$ and $A$ for both the objects in $\Set^T$ and in the base category $\Set$.
	In particular, we use $TC$ to also denote the free $T$-algebra over $C$.
}
The key idea is to find \emph{generators} for $X$.
Our definition of generators is equivalent to the definition of a \emph{scoop} due to Arbib and Manes~\cite[Section~7, Definition~8]{arbib1975_}.

\begin{definition}[Generators for an algebra]
	We say that a set $G$ is a set of generators for a $T$-algebra $X$ whenever there exists a function $g \colon G \to X$ such that $g^\sharp \colon TG \to X$ is a split epi in $\Set$.
\end{definition}
The intuition of requiring a split epi is that every element of $X$ can now be decomposed into a ``combination'' (defined by $T$) of elements of $G$.
We show two simple results on generators, which will allow us to find initial sets of generators for a given $T$-algebra.
\begin{lemma}\label{lem:genself}
	The carrier of any $T$-algebra $X$ is a set of generators for it.
\end{lemma}
\begin{proof}
	Let $TX \xrightarrow{\chi} X$ be the $T$-algebra structure on $X$.
	Then $\id_X$ satisfies $\id_X^\sharp = \chi$, and $\chi$ is a split epi because it is required to satisfy $\chi \circ \eta_X = \id_X$.
\end{proof}

\begin{lemma}\label{lem:genfree}
	Any set $X$ is a set of generators for the free $T$-algebra $TX$.
\end{lemma}
\begin{proof}
	Follows directly from the fact that $\eta_X \colon X \to TX$ satisfies $\eta_X^\sharp = \id_{TX}$.
\end{proof}
Once we have a set of generators $G$ for $X$, we can define an equivalent \emph{free representation} of $\mathcal{X}$, that is, an automaton whose state space is freely generated from $G$.
\begin{proposition}[Free representation of an automaton~{\cite[Section~7, Proposition~9]{arbib1975_}}]\label{prop:free-rep}
	The free algebra $TG$ forms the state space of an automaton equivalent to $\mathcal{X}$.
\end{proposition}
\begin{proof}
	Let $g \colon G \to X$ witness $G$ being a set of generators for $X$ and let $s \colon X \to TG$ be a right inverse of $g^\sharp$.
	Recall that $\mathcal{X} = (X, i, \delta)$ and define
	\begin{align*}
		j &
			= 1 \xrightarrow{i} X \xrightarrow{s} TG \\
		\gamma &
			= G \xrightarrow{g} X \xrightarrow{\delta} O \times X^A \xrightarrow{\id \times s^A} O \times (TG)^A
	\end{align*}
	Then $(TG, j, \gamma^\sharp)$ is an automaton.
	We will show that $g^\sharp \colon TG \to X$ is an automaton homomorphism.
	We have $g^\sharp \circ j = g^\sharp \circ s \circ i = i$, and, writing $F$ for the functor $O \times (-)^A$ and $\chi$ for the $T$-algebra structure on $X$,
	\[
		\begin{tikzcd}[column sep=.7cm,row sep=1.0cm]
			TG \ar{r}{Tg} \ar{dd}[swap]{g^\sharp} &
				TX \ar{r}{T\delta} \ar[bend left=10]{ddl}{\chi} \ar[phantom,bend right,pos=.3]{rrrrdd}{\smallcircled{1}} &
				TFX \ar{r}[pos=.4]{TFs} \ar{rd}[swap]{\rho} &
				TFTG \ar{r}{\rho} \ar[phantom,bend right,pos=.4]{d}{\smallcircled{2}} &
				FT^2G \ar{r}{F\mu} \ar{d}{FTg^\sharp} \ar[phantom,bend left=14,pos=.45]{rdd}{\smallcircled{3}} &
				FTG \ar{dd}{Fg^\sharp} \\
			&
				&
				&
				FTX \ar{r}{\id} \ar{ru}{FTs} &
				FTX \ar{rd}{F\chi} \\
			X \ar{rrrrr}{\delta} &
				&
				&
				&
				&
				FX
		\end{tikzcd}
	\]
	commutes.
	Here $\smallcircled{1}$ commutes because $\delta$ is a $T$-algebra homomorphism, $\smallcircled{2}$ commutes by naturality of the distributive law $\rho$, and $\smallcircled{3}$ commutes because $g^\sharp$ is a $T$-algebra homomorphism.
	The triangle on the left unfolds the definition of $g^\sharp$, and the remaining triangle commutes by $s$ being right inverse to $g^\sharp$.
	Note that the composition in the top row of the diagram is $\gamma^\sharp$.
	We conclude that $g^\sharp$ is an automaton homomorphism, which using the finality in Definition~\ref{def:language} implies that $(TG, j, \gamma^\sharp)$ accepts the same language as $\mathcal{X}$.
\end{proof}

The state space $TG$ of this free representation can be extremely large.
Fortunately, the fact that $TG$ is a free algebra allows for a much more succinct version of this automaton.

\begin{definition}[$T$-succinct automaton]\label{def:succinct}
	Given an automaton of the form $(TX, i, \delta)$, where $TX$ is the free $T$-algebra on $X$, the corresponding \emph{$T$-succinct automaton} is the triple $(X, i, \delta \circ \eta)$.
	The language accepted by the $T$-succinct automaton is the language $l \circ i$ accepted by $(TX, i, \delta)$:
	\[
		\begin{tikzcd}[column sep=.9cm,row sep=.4cm]
			1 \ar{dr}{i} \\
			X \ar{r}{\eta} \ar{dd}[swap]{\delta \circ \eta} &
				TX \ar{ddl}{\delta} \ar[dashed]{r}{l} &
				O^{A^*} \ar{dd}{\omega} \\
			\\
			O \times (TX)^A \ar[dashed]{rr}{\id \times l^A} &
				&
				O \times (O^{A^*})^A
		\end{tikzcd}
	\]
\end{definition}

The goal of our construction is to build a $T$-succinct automaton from a set of generators that is \emph{minimal} in a way that we will define now.
In what follows below we use the following piece of notation: if $U$ and $V$ are sets such that $U \subseteq V$, then we write $\iota^U_V$ for the inclusion map $U \to V$.

\begin{definition}[Minimal generators]\label{def:mingen}
	Given a $T$-algebra $X$ and a set of generators $G$ for $X$ witnessed by $g \colon G \to X$, we say that $r \in G$ is \emph{redundant} if there exists a $U \in T(G \setminus \{r\})$ satisfying $(g \circ \iota^{G \setminus \{r\}}_G)^\sharp(U) = g(r)$; all other elements are said to be \emph{isolated}~\cite{arbib1975_}\footnote{%
		Arbib and Manes~\cite{arbib1975_} define isolated elements only for the full set $X$ rather than relative to a set of generators for $X$.
		Our refinement plays an important role in finding a minimal set of generators.
	}.
	We call $G$ a \emph{minimal set of generators for $X$} if $G$ contains no redundant elements.
\end{definition}

A minimal set of generators is not necessarily minimal in size.
However, under certain conditions this is the case.
The following result was mentioned but not proved by Arbib and Manes~\cite{arbib1975_}, who showed that its conditions are satisfied for any finitely generated $\Ps$-algebra.
We note that these conditions do not apply (in general) to any of the further examples in Section~\ref{sec:examples}.

\begin{proposition}\label{prop:geniso}
	If a $T$-algebra $X$ is generated by the isolated elements $I$ of the set of generators $X$ (Lemma~\ref{lem:genself}) with their inclusion map $\iota^I_X$ and $I$ is finite, then there is no set of generators for $X$ smaller than $I$, and every minimal set of generators for $X$ has the same size as $I$.
\end{proposition}
\begin{proof}
	Let $G \xrightarrow{g} X$ be a set of generators for $X$, and assume towards a contradiction that $G$ is smaller than $I$.
	Then there must be an $i \in I$ such that there is no $v \in G$ satisfying $g(v) = i$.
	Let $g' \colon G \to X \setminus \{i\}$ be pointwise equal to $g$.
	Because $g^\sharp$ is a split epi and thus surjective, there is a $U \in TG$ such that $g^\sharp(U) = i$.
	Note that by \eqref{eq:sharpcomp},
	\[
		g^\sharp = (\iota^{X \setminus \{i\}}_X \circ g')^\sharp = TG \xrightarrow{T(g')} T(X \setminus \{i\}) \xrightarrow{(\iota^{X \setminus \{i\}}_X)^\sharp} X.
	\]
	Then $(\id \circ \iota^{X \setminus \{i\}}_X)^\sharp(T(g')(U)) = i$, contradicting the fact that $i$ is isolated in the full set of generators $X$.
	Thus, $G$ cannot be smaller than $I$.
	In fact, we see that for every $i \in I$ there is a $v \in G$ satisfying $g(v) = i$.
	This yields a function $h \colon I \to G$ such that $g \circ h = \iota^I_X$.

	Suppose $G$ is a minimal set of generators, and take any $v \in G$ not in the image of $h$.
	We will show that $v$ is redundant in $G$.
	Since $I$ constitutes a set of generators for $X$, there exists a $U \in TI$ such that $(\iota^I_X)^\sharp(U) = g(v)$.
	Then
	\[
		g^\sharp(T(h)(U)) = (g \circ h)^\sharp(U) = (\iota^I_X)^\sharp(U) = g(v).
	\]
	It follows that $v$ is redundant in $G$, which contradicts $G$ being minimal.
	Therefore, $h$ is surjective and $G$ has the same size as $I$.
\end{proof}

\section{$T$-minimization}\label{sec:minimization}

In this section we describe a construction to compute a ``minimal'' succinct $T$-automaton equivalent to a given $T$-automaton. This crucially relies on a procedure that finds a minimal set of generators by removing redundant elements one by one. All that needs to be done for specific monads is determining whether an element is redundant.

\begin{proposition}[Generator reduction]\label{prop:reduction}
	Given a $T$-algebra $X$ and a set of generators $G$ for $X$, if $r \in G$ is redundant, then $G \setminus \{r\}$ is a set of generators for $X$.
\end{proposition}
\begin{proof}
	Let $G' = G \setminus \{r\}$ and let $g' \colon G' \to X$ be the restriction of $g \colon G \to X$ to $G'$.
	Since $r$ is redundant, there is a $U \in T(G')$ such that $g'^\sharp(U) = g(r)$.
	Define $e \colon G \to T(G')$ by
	\[
		e(x) = \begin{cases}
			U &
				\text{if $x = r$} \\
			\eta(x) &
				\text{if $x \ne r$.}
		\end{cases}
	\]
	We will show that $g'^\sharp \circ e = g$.
	Consider any $x \in G$.
	If $x = r$, then
	\[
		g'^\sharp(e(x)) = g'^\sharp(e(r)) = g'^\sharp(U) = g(r) = g(x).
	\]
	If $x \ne r$, then, using \eqref{eq:sharpinv},
	\[
		g'^\sharp(e(x)) = g'^\sharp(\eta(x)) = g'^{\sharp\dagger} = g'(x) = g(x).
	\]

	Let $\chi \colon TX \to X$ be the algebra structure on $X$ and take any right inverse $s \colon X \to TG$ of $g^\sharp$.
	Then
	\begin{align*}
		g'^\sharp \circ e^\sharp \circ s &
			= g'^\sharp \circ \mu \circ Te \circ s &
			&
			\text{(definition of $e^\sharp$)} \\
		&
			= \chi \circ T(g'^\sharp) \circ Te \circ s &
			&
			\text{($g'^\sharp$ is a $T$-algebra homomorphism)} \\
		&
			= \chi \circ T(g'^\sharp \circ e) \circ s &
			&
			\text{(functoriality of $T$)} \\
		&
			= \chi \circ Tg \circ s &
			&
			\text{($g'^\sharp \circ e = g$ as shown above)} \\
		&
			= g^\sharp \circ s &
			&
			\text{(definition of $g^\sharp$)} \\
		&
			= \id_X &
			&
			\text{($s$ is right inverse to $g^\sharp$)}.
	\end{align*}
	We thus see that $e^\sharp \circ s$ is right inverse to $g'^\sharp$, which means that $G'$ is a set of generators for $X$.
\end{proof}

If we determine that an element is isolated, there is no need to check this again later when the set of generators has been reduced.
This is thanks to the following result.

\begin{proposition}
	If $G \xrightarrow{g} X$ and $G' \xrightarrow{g'} X$ are sets of generators for a $T$-algebra $X$ such that $G' \subseteq G$ and $g'$ is the restriction of $g$ to the domain $G'$, then whenever an element $r \in G'$ is isolated in $G$, it is also isolated in $G'$.
\end{proposition}
\begin{proof}
	We will show that redundant elements in $G'$ are also redundant in $G$.
	If $r \in G'$ is isolated in $G'$, then there exists $U \in T(G' \setminus \{r\})$ such that $(g' \circ \iota^{G' \setminus \{r\}}_{G'})^\sharp(U) = g'(r)$.
	Note that $g' = g \circ \iota^{G'}_G$.
	We have
	\begin{align*}
		(g \circ \iota^{G \setminus \{r\}}_G)^\sharp(T(\iota^{G' \setminus \{r\}}_{G \setminus \{r\}})(U)) &
			= (g \circ \iota^{G \setminus \{r\}}_G \circ \iota^{G' \setminus \{r\}}_{G \setminus \{r\}})^\sharp(U) &
			&
			\text{\eqref{eq:sharpcomp}} \\
		&
			= (g \circ \iota^{G'}_G \circ \iota^{G' \setminus \{r\}}_{G \setminus \{r\}})^\sharp(U) \\
		&
			= (g' \circ \iota^{G' \setminus \{r\}}_{G \setminus \{r\}})^\sharp(U) \\
		&
			= g'(r) \\
		&
			= g(r),
	\end{align*}
	so $r$ is redundant in $G$.
\end{proof}

Finally, taking the observable quotient $M$ of a $T$-automaton $Q$ preserves generators, considering that the $T$-automaton homomorphism $m \colon Q \to M$ is a split epi in $\Set$ under the axiom of choice.

\begin{proposition}
	If $Q$ and $M$ are $T$-algebras, $m \colon Q \to M$ is a $T$-algebra homomorphism that is a split epi in $\Set$, and $G \xrightarrow{g} Q$ is a set of generators for $Q$, then $G \xrightarrow{g} Q \xrightarrow{m} M$ is a set of generators for $M$.
\label{prop:gen-epi}
\end{proposition}
\begin{proof}
	Let $a \colon TQ \to Q$ be the $T$-algebra structure on $Q$ and $b \colon TM \to M$ the one on $M$.
	We have
	\[
		(m \circ g)^\sharp = b \circ T(m \circ g) = b \circ T(m) \circ T(g) = m \circ a \circ Tg = m \circ g^\sharp
	\]
	using that $m$ is a $T$-algebra homomorphism.
	It is well known that compositions of split epis are split epis themselves, so $G$ is a set of generators for $M$.
\end{proof}

Now we are ready to define the construction that builds a $T$-succinct automaton accepting the same language as a $T$-automaton.

\begin{construction}[$T$-minimization]\label{cons:tmin}
	Starting from a $T$-automaton $(X, i, \delta)$, where $X$ has a finite set of generators, we execute the following steps.
	\begin{enumerate}
		\item
			Take the observable quotient $(M, i_0, \delta_0)$ of $(X, i, \delta)$.
		\item
			Compute a minimal set of generators $G$ of $M$ by starting from the full set $M$ and applying Proposition~\ref{prop:reduction}.
		\item
			Compute and return the corresponding $T$-succinct automaton as defined in Definition~\ref{def:succinct} via Proposition~\ref{prop:free-rep}.
	\end{enumerate}
\end{construction}

Generic minimization algorithms have been proposed in the literature. For example, Ad{\'a}mek et al.\ give a general procedure to compute the observable quotient~\cite{adamek2012}, and K\"onig and K\"upper provide a generic partition refinement algorithm for coalgebras, with a focus on instantiations to weighted automata~\cite{DBLP:conf/ifipTCS/KonigK14}. None of these works provide any complexity analysis. Recently, Dorsch et al.~\cite{DBLP:journals/corr/DorschMSW17} have presented a coalgebraic Paige--Tarjan algorithm and provided a complexity analysis for a class of functors in categories with image-factorization. These restrictions match well the ones we make, and therefore their algorithm could be applied in our first step.
Given a finite set of generators $G$, the loop in the second step involves considering each element of $G$ and checking whether it is redundant.
If so, we will remove the element from $G$ and continue the loop.
The redundancy check is the only part for which computability needs to be determined in each specific setting.

\begin{example} [Join-semilattices]
	We give an example of the construction in the category JSL of complete join-semilattices.
	We start from a minimal $\Ps$-automaton (in JSL) that has 4 states and is depicted below on the left.
	The dashed blue lines indicate the JSL structure.

	\begin{tabular}{c@{\hspace{3cm}}c}
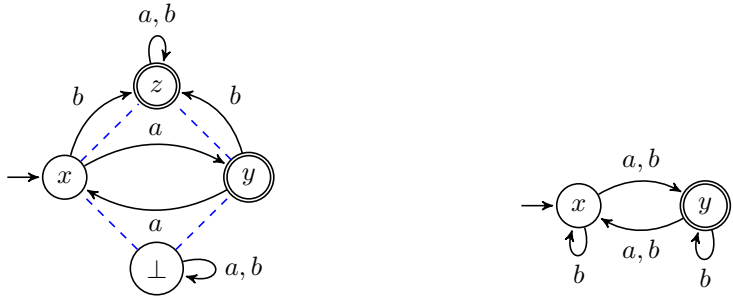
		\begin{tikzpicture}[initial text={},->,>=stealth',shorten >=1pt,auto,node distance=8ex,semithick]
		\node[initial,state] (0) {$x$};
		\node (0a) [right of =0] {};
		\node[state, accepting] (1) [right of = 0a] {$y$};
		\node[state, accepting] (2) [above of = 0a] {$z$};
		\node[state] (3) [below of = 0a] {$\bot$};

		\path
		(0) edge [bend left] node {$b$} (2)
		(1) edge [swap, bend right] node {$b$} (2)
		(2) edge[loop above] node {$a,b$} ()
		(1) edge [bend left] node {$a$} (0)
		(0) edge [bend left] node {$a$} (1)
		(3) edge [-,dashed,blue]  (0)
		(3) edge [-,dashed,blue]  (1)
		(1) edge [-,dashed,blue]  (2)
		(0) edge [-,dashed,blue]  (2)
		(3) edge[loop right] node {$a,b$} ();
		\end{tikzpicture} &
			\begin{tikzpicture}[initial text={},->,>=stealth',shorten >=1pt,auto,node distance=11ex,semithick]
			\node[initial,state] (0) {$x$};
			\node[state,accepting] (1) [right of = 0] {$y$};

			\path
			(0) edge [bend left] node {$a,b$} (1)
			(1) edge [bend left] node {$a,b$} (0)
			(0) edge[loop below] node {$b$} (0)
			(1) edge[loop below] node {$b$} (0);
			\end{tikzpicture}
	\end{tabular}

	Since the automaton is minimal, it is isomorphic to its observable quotient.
	We start from the full set of generators $\{\bot, x, y, z\}$.
	Note that $z$ is the union of $x$ and $y$, so we can eliminate it.
	Additionally, $\bot$ is the empty union and can be removed as well.
	Both $x$ and $y$ are isolated elements and form the unique minimal set of generators $G = \{x, y\}$ (see the remark above Proposition~\ref{prop:geniso}).
	These are exactly the join-irreducibles of $M$.
	They induce by Proposition~\ref{prop:free-rep} an automaton $(TG,j,\gamma)$, where $\gamma$ is the same transition structure as the above automaton, but with $\{x,y\}$ substituted for $z$; the initial state is the singleton set $\{x\}$.
	The $\Ps$-succinct automaton corresponding to this minimal set of generators (Definition~\ref{def:succinct}) is the non-deterministic automaton shown on the right.

	Note that the definition of the automaton defined in Proposition~\ref{prop:free-rep} depends on the right inverse chosen for the extension of the generator map.
	When the original JSL automaton is reachable (every state is reached by some set of words, where a set of words reaches the join of the states reached by the words it contains), this right inverse may be chosen in such a way to recover the canonical \emph{residual finite state automaton} (RFSA), as well as the simplified canonical RFSA, both due to Denis et al.~\cite{denis2002}.
	Details are given in~\cite{vanHeerdtSS18}.
	See \cite{MyersAMU15} for conditions under which the canonical RFSA, referred to as the jiromaton, is a state-minimal NFA.
\end{example}

\section{Main construction}\label{sec:main}

In this section we present the main construction of the paper.
Given a \emph{finite} automaton $(X,i,\delta)$ in $\Set$, i.e., an automaton where $X$ is finite, this construction builds an equivalent $T$-succinct automaton.

The first step is taking the reachable part $R$ of $X$ and converting this automaton into a $T$-automaton recognising the same language.
\begin{proposition}\label{prop:undet}
	Let $(TR, \hat{i}, \hat{\delta})$ be the $T$-automaton defined as follows:
	\[
		\begin{tikzcd}
			1 \ar{d}{i} \ar[dashed]{dr}{\hat{i} = \eta_R \circ i} \\
			R \ar{r}{\eta_R} \ar{d}{\delta} &
				TR \ar[dashed]{d}{\hat{\delta} = ((\id \times \eta_R^A) \circ \delta)^\sharp} \\
			O \times R^A \ar{r}{\id \times \eta_R^A} &
				O \times T(R)^A
		\end{tikzcd}
	\]
	Then $(R,i,\delta)$ and $(TR, \hat{i}, \hat{\delta})$ accept the same language.
\end{proposition}
\begin{proof}
	The diagram above means that $\eta_R$ is a coalgebra homomorphism, and as such it preserves language.
	Explicitly: $x \in R$ accepts the same language as $\eta_R(x)$, which in particular holds for $i(\star)$ and $\hat{i}(\star)$.
\end{proof}
Now we can $T$-minimize $(TR, \hat{i}, \hat{\delta})$ (Construction~\ref{cons:tmin}), which yields an equivalent $T$-automaton.
Notice that, $R$ being finite, any quotient of $TR$ has a finite set of generators. This is a consequence of $R$ being a set of generators for $TR$ (Lemma~\ref{lem:genfree}) and of generators being preserved by quotients (Proposition~\ref{prop:gen-epi}). It follows that every step of the $T$-minimization construction terminates.

\begin{proposition}
	The $T$-succinct automaton defined above is at least as small as the minimal deterministic automaton equivalent to $X$.
\end{proposition}
\begin{proof}
	The situation is summed up in the following commutative diagram:
	\[
		\begin{tikzcd}[column sep=.6cm,row sep=.7cm]
			G \ar[hook]{d} \ar[bend left=20]{drr}{g} \\
			R \ar{r}{\eta} &
				TR \ar{r}{e} &
				M \ar[tail]{r}{m} &
				O^{A^*}
		\end{tikzcd}
	\]
	Here $G$ is the final minimal set of generators for $M$ resulting from the construction.
	Commutativity follows from $G$ being a subset of the set of generators $R$.

	The minimal deterministic automaton equivalent to $X$ is obtained from $R$ by merging language-equivalent states.
	Recalling (\ref{eq:obs}) and the proof of Proposition~\ref{prop:undet}, we see that $e \circ \eta_R$ is a coalgebra homomorphism.
	Together with commutativity of the above diagram, this means that the language accepted by $r \in G$ (seen as a state of $R$) is given by $(m \circ g)(r)$.
	Since $G$ is a subset of $R$, to show that $G$ is at least as small as the minimal deterministic automaton, we only have to show that different states in $G$ accept different languages.
	That is, we will show that $m \circ g$ is injective.
	We know that $m$ is injective by definition; to see that $g$ is injective, consider $r_1, r_2 \in G$ such that $g(r_1) = g(r_2)$.
	Then $g(r_1) = g(r_2) = g^\sharp(\eta(r_2))$.
	Assuming $r_1 \ne r_2$ leads to the contradiction that $G$ is not a minimal set of generators because in this case $\eta(r_2) \in T(G \setminus \{r_1\})$.
\end{proof}

Computing the determinization $TR$ is an expensive operation that only terminates if $T$ preserves finite sets. One could devise an optimized
version of Construction~\ref{cons:tmin} in which the determinization is not computed completely in order to minimize it. Instead, we could choose to work with data structures as B\"ollig et al.~\cite{bollig2009} did for non-deterministic automata, and which we generalized in recent work~\cite{vanHeerdtSS18}. In these papers, partial representations of the determinized automaton are used in an iterative process to compute the generators of the state space of the minimal one. 
\section{Examples}\label{sec:examples}

\subsection{Monads preserving finite sets}\label{sec:finite}

If $T$ preserves finite sets, then there is a naive method to find a redundant element: assuming a finite set of generators $G$ for a $T$-algebra $X$, the set $T(G \setminus \{r\})$ is also finite for any $r \in G$.
Thus, we can loop over all $U \in T(G \setminus \{r\})$ and check if the generator map $g \colon G \to X$ satisfies $g^\sharp(U) = g(r)$.

\subsubsection{Alternating automata.}
\label{sec:afa}
\newcommand{\upclos}{\uparrow\!\!}

We now use our construction to get small alternating finite automata (AFAs) over a finite alphabet $A$.
AFAs generalize both non-deterministic and universal automata, where the latter are the dual of non-deterministic automata: a word is accepted when all paths reading it are accepting.
In an AFA, reading a symbol leads to a DNF formula (without negation) of next states.

We use the characterization of alternating automata due to Bertrand~\cite{BertrandR18}.
Given a partially ordered set $(P,\leq)$, an \emph{upset} is a subset $U$ of $P$ such that whenever $x \in U$ and $x \leq y$, then $y \in U$.
Given $Q \subseteq P$, we write $\upclos Q$ for the \emph{upward closure} of $Q$, that is the smallest upset of $P$ containing $Q$.
We consider the monad $\TAlt$ that maps a set $X$ to the set of all upsets of $\Ps(X)$.
Its unit is given by $\eta_X(x) = \upclos \{\{ x \}\}$ and its multiplication by
\[
	\mu_X(U) = \{V \subseteq X \mid \exists_{W \in U}\,\forall_{Y \in W}\,\exists_{Z \in Y}\,Z \subseteq V\}.
\]
The sets of sets in $\TAlt(X)$ can be seen as DNF formulae over elements of $X$: the outer powerset is interpreted disjunctively and the inner one conjunctively.
Accordingly, we define an algebra structure $\beta \colon \TAlt(2) \to 2$ on the output set $2$ by letting $\beta(U) = 1$ if $\{ 1 \} \in U$ and $\beta(U) = 0$ otherwise.
Recall from (\ref{eq:stddl}) in Section~\ref{sec:preliminaries} that such an algebra structure induces a distributive law.

We now explicitly spell out the $T$-minization algorithm that turns a DFA $(X,i,\delta)$ into a $\TAlt$-succinct $AFA$.
\begin{enumerate}
	\item Compute the reachable states $R$ of $(X,i,\delta)$ via a standard visit of its graph.
	\item Compute the corresponding freely-generated $\TAlt$-automaton $(\TAlt R,\hat{i},\hat{\delta})$, by generating all DNF formulae $\TAlt R$ on $R$.
	\item Compute the observable quotient $(M,i_0,\delta_0)$ of $(\TAlt R,\hat{i},\hat{\delta})$ via a standard minimization algorithm, such as the coalgebraic Paige--Tarjan algorithm~\cite{DBLP:journals/corr/DorschMSW17}.
	\item Compute a minimal set of generators for $M$ as follows. Consider the generator map $\id_M \colon M \to M$, for which we have that $\id^{\sharp}$ is the algebra map of $M$. Pick $r \in M$, and iterate over all DNF formulae $\varphi$ over $M \setminus \{r\}$; if there is $\varphi$ which is mapped to $r$ by the algebra map of $M$ (i.e., $\id^{\sharp}$), $r$ is redundant and can be removed from $M$. Repeat until no more elements are removed from $M$, which yields a minimal set of generators $G$.
	\item Return the $\TAlt$-succinct automaton $(G,i_0,i_0 \circ \eta)$.
\end{enumerate}
Note that every step of this algorithm terminates, as $X$ is finite and the size of $\TAlt R$ is $2^{2^{|R|}}$.

\begin{example}
	\begin{figure}
		\centering
		\subfloat[Deterministic automaton]{\label{fig:aa1}
			\centering
			\begin{tikzpicture}[initial text={},->,>=stealth',shorten >=1pt,auto,node distance=14ex,semithick]
				\node[initial,state] (0) {$q_0$};
				\node[state] (1) [below of=0] {$q_1$};
				\node[state,accepting] (2) [right of=0] {$q_2$};
				\node[state,accepting] (3) [below of=2] {$q_3$};
				\node[state] (4) [right of=2] {$q_4$};
				\path
				(0) edge node {$a$} (2)
				(0) edge node [swap] {$b$} (1)
				(1) edge node {$a$} (3)
				(1) edge node {$b$} (2)
				(2) edge node {$a,b$} (4)
				(3) edge node {$a$} (2)
				(3) edge node {$b$} (4)
				(4) edge [loop below] node {$a,b$} ();
			\end{tikzpicture}
		}\qquad%
		\subfloat[Small corresponding AFA]{\label{fig:aa2}
			\centering
			\begin{tikzpicture}[initial text={},->,>=stealth',shorten >=1pt,auto,node distance=14ex,semithick,square/.style={regular polygon,regular polygon sides=4}]
				\node[initial,state] (0) {$q_0$};
				\node[state,accepting] (1) [right of=0] {$q_2$};
				\node[square,draw,fill] (2) [below of=0] {};
				\node[state] (3) [below of=1] {$q_1$};
				\path
				(0) edge node {$a$} (1)
				(0) edge node {$b$} (3)
				(2) edge node {} (0)
				(2) edge [bend right] node {} (3)
				(3) edge node [swap] {$a$} (2)
				(3) edge node [swap] {$a,b$} (1);
			\end{tikzpicture}
		}
		\caption{Automata for the language $\{a, ba, bb, baa\}$}
	\end{figure}
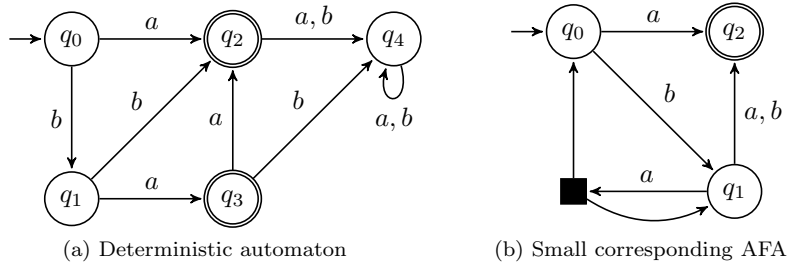
	Consider the regular language over $A = \{a, b\}$ given by the finite set $\{a, ba, bb, baa\}$.
	The minimal DFA accepting this language is given in \figurename~\ref{fig:aa1}.

	According to our construction, we first construct a $\TAlt$-automaton with state space freely generated from this automaton (which is already reachable). Then we $\TAlt$-minimize it in order to obtain a small AFA. 
	In this case, there is a unique minimal subset of $3$ generators: $G = \{q_0, q_1, q_2\}$.
	To see this, consider the languages $\llbracket{}q\rrbracket$ accepted by states $q$ of the deterministic automaton:
	\begin{align*}
		\llbracket{}q_0\rrbracket{} &
			= \{a, ba, bb, baa\} &
			\llbracket{}q_2\rrbracket{} &
			= \{\eword\} &
			\llbracket{}q_4\rrbracket{} &
			= \emptyset \\
		\llbracket{}q_1\rrbracket{} &
			= \{a, b, aa\} &
			\llbracket{}q_3\rrbracket{} &
			= \{\eword, a\}.
	\end{align*}
	These languages generate the states of the minimal $\TAlt$-automaton by interpreting joins as unions and meets as intersections.
	We note that $\llbracket{}q_4\rrbracket$ is just an empty join and $\llbracket{}q_3\rrbracket{} = (\llbracket{}q_0\rrbracket{} \cap \llbracket{}q_1\rrbracket{}) \cup \llbracket{}q_2\rrbracket{}$.\footnote{%
		Strictly speaking, we should take the upwards-closure of this disjunction (adding any possible set of elements to each conjunction as an additional clause).
		We choose to use the equivalent succinct formula both here and in the subsequent AFA construction to aid readability.
	}
	These are the only redundant generators.
	Removing them leads to the AFA in Figure~\ref{fig:aa2}.
	Here the black square represents a conjunction of next states.
\end{example}

\subsubsection{Complete Atomic Boolean Algebras}

We now consider the monad $C$ given by the double \emph{contravariant} powerset functor, namely $CX = 2^{2^X}$.
Here the outer powerset is treated disjunctively as in the case of $\TAlt$, and the sets provided by the inner powerset are interpreted as valuations.
Thus, elements of $C(X)$ can be seen as \emph{full} DNF formulae over $X$: every conjunctive clause contains for each $x \in X$ either $x$ or the negation $\overline{x}$ of $x$.
The unit assigns to an element $x$ the disjunction of all full conjunctions containing $x$, and the multiplication turns formulae of formulae into full DNF formulae in the usual way.
Algebras for this monad are known as complete atomic boolean algebras (CABAs).

Using the fact that $2$ is a free CABA ($2 \cong C(\emptyset)$), we obtain the following semantics for $C$-succinct automata: a set of sets of states is accepting if and only if it contains the exact set $F$ of accepting states.
This is different from alternating automata, where a subset of $F$ is sufficient.
Reading a symbol in a $C$-succinct automaton works as follows.
Suppose we are in a set of sets of states $S \in C(Q)$, where we read a symbol $a$.
The resulting set of sets contains $U \subseteq Q$ if and only if there is a set $V \in S$ such that every state in $V$ transitions into a set of sets containing $U$, and every state not in $V$ does not transition into any set of sets containing $U$.

Note that every DNF formula can be converted to a full DNF formula.
This implies that $C$-succinct automata can always be as small as the smallest AFAs for a given language.
With the following example we show that they can actually be strictly smaller. The $T$-minimization algorithm for AFA we have given in the previous section applies to this setting as well (including negation in DNF formulae).

\begin{example}
	\begin{figure}
		\centering
		\subfloat[Deterministic automaton]{\label{fig:ca1}
			\centering
			\begin{tikzpicture}[initial text={},->,>=stealth',shorten >=1pt,auto,node distance=14ex,semithick,square/.style={regular polygon,regular polygon sides=4}]
				\node[initial,state] (0) {$q_0$};
				\node[state] (1) [below of=0] {$q_1$};
				\node[state,accepting] (2) [right of=1] {$q_2$};
				\path
				(0) edge [swap] node {$a$} (1)
				(1) edge [bend left] node {$a$} (2)
				(2) edge [bend left] node {$a$} (1);
			\end{tikzpicture}
		}\qquad%
		\subfloat[$C$-succinct automaton]{\label{fig:ca2}
			\centering
			\begin{tikzpicture}[initial text={},->,>=stealth',shorten >=1pt,auto,node distance=14ex,semithick,square/.style={regular polygon,regular polygon sides=4}]
				\node[initial,state] (0) {$q_0$};
				\node[state] (1) [below of=0] {$q_1$};
				\node[square,draw,fill] (2) [right of=1] {};
				\path
				(0) edge [bend left] node {$a$} (1)
				(1) edge [bend left] node {$a$} (0)
				(1) edge node {$a$} (2);
			\end{tikzpicture}
		}
		\caption{Automata for the language of non-zero even words over $\{a\}$}
	\end{figure}
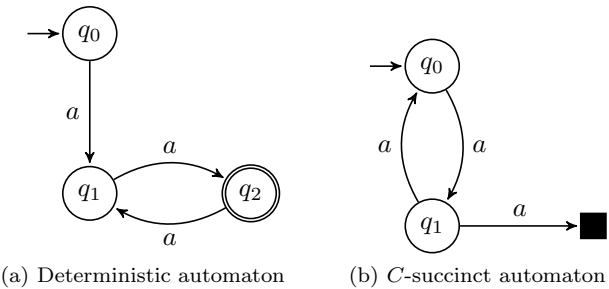
	Consider the regular language of words over the singleton alphabet $A = \{a\}$ whose length is non-zero and even.
	The minimal DFA accepting this language is shown in \figurename~\ref{fig:ca1}.
	We start the algorithm with the $C$-automaton with state space freely generated from this this DFA and merge the language-equivalent states.
	Initially, the set of generators is the set of states of the original DFA.
	By noting that the language accepted by $q_2$ is the negation of the one accepted by $q_1$, in full DNF form $\llbracket{}q_2\rrbracket{} = (\llbracket{}q_0\rrbracket{} \cap \overline{\llbracket{}q_1\rrbracket{}}) \cup (\overline{\llbracket{}q_0\rrbracket{}} \cap \overline{\llbracket{}q_1\rrbracket{}})$ (where for any language $U$ its complement is defined as $\overline{U} = A^* \setminus U$), we see that $q_2$ is redundant.
	The set of generators $\{q_0, q_1\}$ is minimal and corresponds to the $C$-succinct automaton in \figurename~\ref{fig:ca2}.
	We depict $C$-succinct automata in the same manner as AFAs, but note that their interpretation is different.
	Here the transition into the black square represents the transition into the conjunction of the negations of $q_0$ and $q_1$.

	We now show that there is no AFA with two states accepting the same language.
	Suppose such an AFA exists, and let the state space be $X = \{x_0, x_1\}$.
	Since $a$ and $aaa$ are not in the language but $aa$ is, one of these states must be accepting and the other must be rejecting.\footnote{%
		If there were no rejecting states, the only way to reject a word is by ending up in the empty set of sets of states.
		However, this means that extensions of that word are rejected as well.
		Similarly, if there are no accepting states one can only accept by ending up in $\upclos\{\emptyset\}$, which accepts everything.
	}
	Without loss of generality we assume that $x_0$ is rejecting and $x_1$ is accepting.
	The empty word is not in the language, so our initial configuration has to be $\upclos\{\{x_0\}\}$.
	Since $a$ is also not in the language, $x_0$ will have to transition to $\upclos\{\{x_0\}\}$ as well.
	However, this implies that $aa$ is not accepted by the AFA, which contradicts the assumption that it accepts the right language.
\end{example}

Unfortunately, the fact that the transition behavior of a set of states depends on states not in that set generally makes it difficult to work with $C$-succinct automata by hand.

\subsubsection{Symmetry}

We now consider succinct automata that exploit symmetry present in their accepted language.
Given a finite group $G$, consider the monad $G \times (-)$, where the unit pairs any element with the unit of $G$ and the multiplication applies the multiplication of $G$.
The algebras for $G \times (-)$ are precisely left group actions.
We assume an action on the alphabet $A$; if no such action is relevant, one may consider the trivial action $G \times A \xrightarrow{\pi_2} A$.
We also assume an action on the output set $O$.
Group actions will be denoted by a centered dot.
We consider the distributive law $\rho \colon G \times (O \times (-)^A) \nto O \times (G \times (-))^A$ given by
\[
	\rho_X(g, o, f) = (g \cdot o, \lambda a. (g, f(g^{-1} \cdot a))).
\]
We explain the resulting semantics of $(G \times (-))$-succinct automata in an example.

\begin{example}
	Consider the group $\mathtt{Perm}(\{a, b\}) = \{e, (ab)\}$ of permutations over elements $a$ and $b$.
	Here $e$ is the identity and $(ab)$ swaps $a$ and $b$.
	We consider the alphabet $A = \{a, b\}$ with an action $\mathtt{Perm}(A) \times A \to A$ given by applying the permutation to the element of $A$, and the output set $O = A \cup \{\bot\}$ with an action given by
	\begin{align*}
		(ab) \cdot a &
			= b &
			(ab) \cdot b &
			= a &
			(ab) \cdot \bot &
			= \bot.
	\end{align*}

	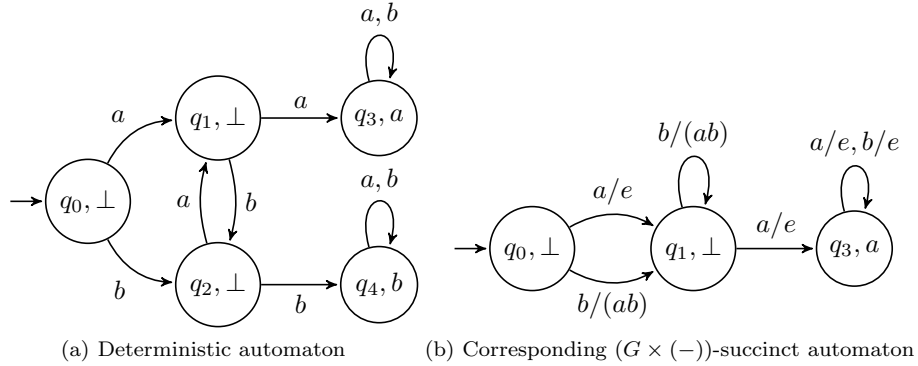
\begin{figure}
		\centering
		\subfloat[Deterministic automaton]{\label{fig:ga1}
			\centering
			\begin{tikzpicture}[initial text={},->,>=stealth',shorten >=1pt,auto,node distance=14ex,semithick]
				\node[initial,state] (0) {$q_0, \bot$};
				\node[state] (1) [above right=2ex and 6ex of 0] {$q_1, \bot$};
				\node[state] (2) [below right=2ex and 6ex of 0] {$q_2, \bot$};
				\node[state] (3) [right of=1] {$q_3, a$};
				\node[state] (4) [right of=2] {$q_4, b$};
				\path
				(0) edge [bend left] node {$a$} (1)
				(0) edge [bend right] node [swap] {$b$} (2)
				(1) edge node {$a$} (3)
				(1) edge [bend left=15] node {$b$} (2)
				(2) edge node [swap] {$b$} (4)
				(2) edge [bend left=15] node {$a$} (1)
				(3) edge [loop above] node {$a,b$} ()
				(4) edge [loop above] node {$a,b$} ();
			\end{tikzpicture}
		}%
		\subfloat[Corresponding $(G \times (-))$-succinct automaton]{\label{fig:ga2}
			\centering
			\begin{tikzpicture}[initial text={},->,>=stealth',shorten >=1pt,auto,node distance=14ex,semithick]
				\node[initial,state] (0) {$q_0, \bot$};
				\node[state] (1) [right of=0] {$q_1, \bot$};
				\node[state] (2) [right of=1] {$q_3, a$};
				\path
				(0) edge [bend left] node {$a/e$} (1)
				(0) edge [bend right] node [swap] {$b/(ab)$} (1)
				(1) edge node {$a/e$} (2)
				(1) edge [loop above] node {$b/(ab)$} ()
				(2) edge [loop above] node {$a/e,b/e$} ();
			\end{tikzpicture}
		}
		\caption{Automata outputting the first symbol to appear twice in a row}
	\end{figure}

	\figurename~\ref{fig:ga1} shows a deterministic automaton over the alphabet $A$ with outputs in $O$.
	States are labeled by pairs $(q, o)$, where $q$ is a state label and $o$ the output of the state.
	The recognized language is the one assigning to a word over $A$ the first input symbol appearing twice in a row, or $\bot$ if no such symbol exists.
	This deterministic automaton is in fact the minimal $(\mathtt{Perm}(A) \times (-))$-automaton.
	The action on its state space is defined by
	\begin{align*}
		(ab) \cdot q_0 &
			= q_0 &
			(ab) \cdot q_1 &
			= q_2 &
			(ab) \cdot q_2 &
			= q_1 &
			(ab) \cdot q_3 &
			= q_4 &
			(ab) \cdot q_4 &
			= q_3.
	\end{align*}
	We note that in the set of generators given by the full state space, $q_1$, $q_2$, $q_3$, and $q_4$ are redundant.
	After removing $q_2$, only $q_3$ and $q_4$ are redundant.
	Subsequently removing $q_4$ leaves no redundant elements.
	
	The final $(G \times (-))$-succinct automaton is shown in \figurename~\ref{fig:ga2}.
	Its actual configurations are pairs of a group element and a state.
	Transition labels are of the form $x/g$, where $x \in A$ and $g \in \mathtt{Perm}(A)$.
	If we are in a configuration $(g, q)$ and state $q$ has an associated output $o \in O$, the actual output is $g \cdot o$.
	On reading a symbol $x \in A$, we find the outgoing transition of which the label starts with the symbol $g^{-1} \cdot x$.
	Supposing this label contains a group element $g'$ and leads to a state $q'$, the resulting configuration is $(gg', q')$.
	For example, consider reading the word $bb$.
	We start in the configuration $(e, q_0)$.
	Reading $b$ here simply takes the transition corresponding to $b$, which brings us to $((ab), q_1)$.
	Now reading the second $b$, we actually read $(ab)^{-1} \cdot b = (ab) \cdot b = a$.
	This brings us to $((ab), q_3)$.
	The output is then given by $(ab) \cdot a = b$.
\end{example}

In general, sets of generators in this setting correspond to subsets in which all \emph{orbits} are represented.
The orbits of a set $X$ with a left group action are the equivalence classes of the relation that identifies elements $x, y \in X$ whenever there exists $g \in G$ such that $g \cdot x = y$.
Minimal sets of generators contain a single representative for each orbit. The algorithm given for AFAs in section~\ref{sec:afa} can be applied to this setting as well: step 4 will remove elements until only orbit representatives are left.

\subsection{Vector Spaces}\label{sec:vector}

We now exploit vector space structures.
Given a field $\F$, consider the free vector space monad $V$.
It maps each set $X$ to the set of functions $X \to \F$ with finite support (finitely many elements of $X$ are mapped to a non-zero value).
A function $f \colon X \to Y$ is mapped to the function $V(f) \colon V(X) \to V(Y)$ given by
\[
	V(f)(g)(y) = \sum_{x \in X, f(x) = y} g(x).
\]
The unit $\eta \colon X \to V(X)$ and multiplication $\mu \colon VV(X) \to V(X)$ of the monad are given by
\begin{align*}
	\eta(x)(x') = \begin{cases}
		1 & \text{if $x = x'$} \\
		0 & \text{if $x \ne x'$}
	\end{cases} &
		&
		\mu(f)(x) = \sum_{g \in V(X)} f(g) \cdot g(x) \in \F.
\end{align*}
Here $0$ and $1$, as well as addition and multiplication, are those of the field $\F$.
Elements of $V(X)$ can alternatively be written as formal sums $v_1x_1 + \cdots + v_nx_n$ with $v_i \in \F$ and $x_i \in X$ for all $i$.
We will use this notation in the example below.

Algebras for the free vector space monad are precisely vector spaces.
We use the output set $O = \F$, and the alphabet can be any finite set $A$.
Instantiating (\ref{eq:stddl}), this leads to a pointwise distributive law $\rho \colon V(O \times (-)^A) \nto O \times V(-)^A$ given at a set $X$ by
\[
	\rho(f) = \left(\sum_{(o, g) \in O \times X^A} f(o, g) \cdot o, \lambda a. \lambda x. \sum_{(o, g) \in O \times X^A, g(a) = x} f(o, g)\right).
\]
With these definitions, the $V$-succinct automata are weighted automata.
We note that if $\F$ is infinite, any non-trivial $V$-automaton will also be infinite.
However, we can still start from a given weighted automaton and apply a slight modification of Construction~\ref{cons:tmin}: minimize from the succinct representation, use the states of the succinct representation as initial set of generators, and finally find a minimal set of generators.
Moreover, we may add a reachability analysis, which in this case cannot lead to a larger automaton.
Thus, the resulting algorithm essentially comes down to the standard minimization algorithm for weighted automata~\cite{schutzenberger1961}, where the process of removing redundant generators is integrated into the minimization.
If $\F$ is finite and we do want to start from a deterministic automaton, we can consider this automaton as a weighted one by assigning each transition a weight of 1.

\begin{example}
	\begin{figure}
		\centering
		\subfloat[Deterministic automaton]{\label{fig:wa1}
			\centering
			\begin{tikzpicture}[initial text={},->,>=stealth',shorten >=1pt,auto,node distance=14ex,semithick]
				\node[initial,state] (0) {$q_0, 0$};
				\node[state] (1) [above right of=0] {$q_1, 1$};
				\node[state] (2) [right of=0] {$q_2, 1$};
				\node[state] (3) [below right of=0] {$q_3, 3$};
				\node[state] (4) [right of=2] {$q_4, 0$};
				\path
				(0) edge node {$a$} (1)
				(0) edge node {$b$} (2)
				(0) edge node [swap] {$c$} (3)
				(1) edge node {$a,b,c$} (2)
				(3) edge node [swap] {$a,b,c$} (2)
				(2) edge node {$a,b,c$} (4)
				(4) edge [loop above] node {$a,b,c$} ();
			\end{tikzpicture}
		}\qquad%
		\subfloat[Succinct weighted automaton]{\label{fig:wa2}
			\centering
			\begin{tikzpicture}[initial text={},->,>=stealth',shorten >=1pt,auto,node distance=14ex,semithick]
				\node[initial,state] (0) {$q_0, 0$};
				\node[state] (1) [above right of=0] {$q_1, 1$};
				\node[state] (2) [below right of=0] {$q_2, 1$};
				\path
				(0) edge node {$a,c$} (1)
				(0) edge node [swap] {$b,c/2$} (2)
				(1) edge [bend left] node {$a,b,c$} (2);
			\end{tikzpicture}
		}
		\caption{Succinctness via a weighted automaton}
	\end{figure}
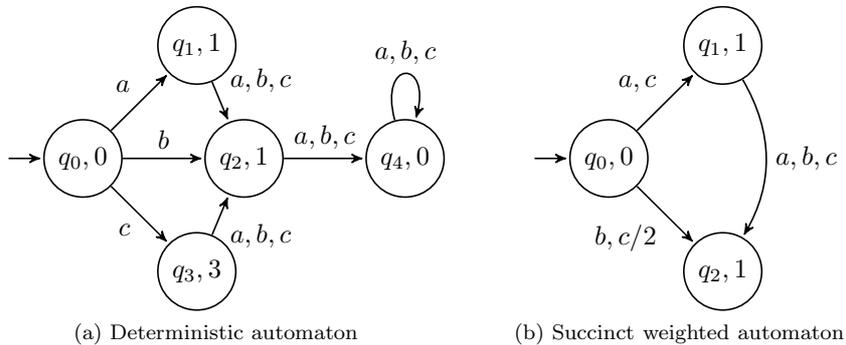

	Consider for $\F = \mathbb{R}$ the deterministic automaton in \figurename~\ref{fig:wa1}.
	This is a minimal automaton in $\Set$; the freely generated $V$-automaton is infinite, and so is its minimization.
	However, that minimization has the states of the automaton in \figurename~\ref{fig:wa1} as a set of generators.
	To gain insight into this minimization, we compute the languages accepted by those generators (apart from $q_0$):
	\begin{align*}
		&
			q_1 \colon &
			\eword &
			\mapsto 1 &
			a &
			\mapsto 1 &
			b &
			\mapsto 1 &
			c &
			\mapsto 1 \\
		&
			q_2 \colon &
			\eword &
			\mapsto 1 &
			a &
			\mapsto 0 &
			b &
			\mapsto 0 &
			c &
			\mapsto 0 \\
		&
			q_3 \colon &
			\eword &
			\mapsto 3 &
			a &
			\mapsto 1 &
			b &
			\mapsto 1 &
			c &
			\mapsto 1 \\
		&
			q_4 \colon &
			\eword &
			\mapsto 0 &
			a &
			\mapsto 0 &
			b &
			\mapsto 0 &
			c &
			\mapsto 0
	\end{align*}
	Words not displayed are mapped to $0$ by any state.
	The language of $q_0$ is the only one assigning non-zero values to certain words of length two, such as $aa$, and therefore $q_0$ cannot be a redundant generator.
	The other generators \emph{are} redundant: writing $\llbracket{}q\rrbracket$ for the language of a state $q$, $\llbracket{}q_4\rrbracket$ is just a zero-ary sum, and we have
	\begin{align*}
		\llbracket{}q_1\rrbracket &
			= \llbracket{}q_3\rrbracket - 2\llbracket{}q_2\rrbracket &
		\llbracket{}q_2\rrbracket &
			= \frac{1}{2}\llbracket{}q_3\rrbracket - \frac{1}{2}\llbracket{}q_1\rrbracket &
		\llbracket{}q_3\rrbracket &
			= \llbracket{}q_1\rrbracket + 2\llbracket{}q_2\rrbracket.
	\end{align*}
	Once $q_4$ is removed, all other generators are still redundant.
	Further removing $q_3$ makes $q_1$ and $q_2$ isolated.
	Therefore, $V$-minimization yields the weighted automaton shown in \figurename~\ref{fig:wa2}.
	Here a transition on an input $x \in A$ with weight $w \in \F$ receives the label $x/w$, or just $x$ if $w = 1$.
	Weights multiply along a path, and different possible paths add up to assign a value to a word.
	Reading $c$ from $q_0$, for example, we move to $q_1 + 2q_2$, which has an output of $1 + 2 * 1 = 3$.
\end{example}

In general, the (sub)sets of generators of a vector space are its subsets that span the whole space, and such a set of generators is minimal precisely when it forms a basis.
The weighted automaton resulting from our algorithm is the usual minimal weighted automaton for the language. 
Redundant elements can be found using standard techniques such as Gaussian elimination.

\section{Conclusions}\label{sec:conclusions}

We have presented a construction to obtain succinct representations of deterministic finite automata as automata with side-effects.
This construction is very general in that it is based on the abstract characterisation of side-effects as monads.
Nonetheless, it can be easily implemented.
An essential part of our construction is the computation of a minimal set of generators for an algebra.
We have provided an algorithm for this that works for any suitable $\Set$ monad.
We have applied the construction to several non trivial examples: alternating automata, automata with symmetries, CABA-structured automata, and weighted automata.

\paragraph{Related work}
This work revamps and extends results of Arbib and Manes~\cite{arbib1975_}, as discussed throughout the paper.
We note that most of their results are formulated in a more general category, whereas here we work specifically in $\Set$.
The reason for this is that we focus on the procedure for finding minimal sets of generators by removing redundant elements, which are defined using set subtraction (Definition~\ref{def:mingen}).
This limitation is already present in the work of Arbib and Manes, who spend little time on the subject and only study the non-deterministic case in detail.
Our main contribution, the general procedure for finding a minimal set of generators, is not present in their work.
It generalizes several techniques to obtain compact automaton representations of languages, some of them presented in the context of learning algorithms~\cite{denis2002,bollig2009,angluin2015}.
Preliminary results on generalizing succinct automaton constructions within a learning algorithm can be found in~\cite{vanHeerdtSS18}.

In \cite{MyersAMU15}, Myers et al.\ present a coalgebraic construction of canonical non-deterministic automata.
Their specific examples are the \'atomaton~\cite{brzozowski2011}, obtained from the atoms of the boolean algebra generated by the residual languages (the languages accepted by the states of the minimal DFA); the canonical RFSA; the minimal xor automaton~\cite{vuillemin2010}, actually a weighted automaton over field with two elements rather than a non-deterministic one; and what they call the distromaton, obtained from the atoms of the distributive lattice generated by the residual languages.
They further provide specific algorithms for obtaining some of their example succinct automata.

The underlying idea in the work of Myers et al.\ for finding succinct representations of algebras is similar to ours, and the deterministic structured automata they start from are equivalent: in their paper the deterministic automata live in a locally finite variety, which translates to the category of algebras for a monad that preserves finite sets (such as those in Section~\ref{sec:finite}).
They also define the succinct automaton using a minimal set of generators for the algebra, but instead of our algorithmic approach of getting to this set by removing redundant generators, they use a dual equivalence between finite algebras and a suitable modification of the category of sets and relations between them.
This seems to restrict their work to non-deterministic automata, although there may be an easy generalization: the equivalence would be with a modification of a Kleisli category.
A major difference with our work is that they have no general algorithm to construct the succinct automata; as mentioned, specific ones are provided for their examples. In fact, they provide no guidelines on how to find a suitable equivalence for a given variety.
On the other hand, their equivalences guarantee uniqueness up to isomorphism of the succinct automata, which is a desirable property for many applications.

The restriction in the work of Myers et al.\ to locally finite varieties means that our example of weighted automata over an infinite field (Section~\ref{sec:vector}) cannot be captured in their work.
Conversely, since both the \'atomaton and the distromaton are non-deterministic NFAs obtained from categories of algebras with more structure than JSLs, these examples are not covered by our work.
Their other examples, however, the canonical RFSA and the minimal xor automaton, are obtained using instances of our method as well.
The fact that the problem of finding in general a suitable equivalence is open means it is not trivial to determine whether our approach can be seen as a special case of a generalized version of theirs when we restrict to monads that preserve finite sets.

\paragraph{Future work}
The main question that remains is under which conditions the notion of a minimal set of generators actually describes a size-minimal set of generators.
Proposition~\ref{prop:geniso} provides a partial answer to this question, but its conditions fail to apply the majority of our examples, even though in some of these cases minimal does mean size-minimal.
A related question is whether we can find heuristics to increase the state space of a $T$-automaton in such a way that the number of generators decreases.
The reason the canonical RFSAs of Denis et al.~\cite{denis2002} are not always state-minimal NFAs is because the states of these NFAs, seen as singletons in the determinized automaton, in general are not reachable.
Hence, removing unreachable states from a $T$-automaton may increase the size of minimal sets of generators, which is why Construction~\ref{cons:tmin} does not include a reachability analysis.
Although finding state-minimal NFAs is PSPACE-complete, a moderate gain might still be possible.

\bibliographystyle{plain}
\bibliography{main}

\end{document}